\documentclass[lettersize,journal]{IEEEtran}

\usepackage{amsmath}
\usepackage{amsfonts}
\usepackage{amssymb}
\usepackage{amsthm}
\usepackage{tabularx}
\usepackage{mathrsfs} 

\usepackage{url}
\usepackage{hyperref}
\newtheorem{theorem}{Theorem}
\newtheorem{remark}{Remark}

\usepackage{stfloats}
\usepackage{float}
\usepackage{graphicx}
\usepackage{cite}
\usepackage{xcolor}
\usepackage{subfigure}

\renewcommand{\vec}[1]{\mathbf{#1}}

\makeatletter
\def\blfootnote{\xdef\@thefnmark{}\@footnotetext}
\makeatother

\begin{document}
	% Needs to be substantially different to: Coverage Characterization of STAR-RIS Networks: NOMA and OMA
		\title{\huge{Analytical Characterization of Coverage Regions for STAR-RIS-aided NOMA/OMA Communication Systems}} 
	\author{Farshad~Rostami~Ghadi\IEEEmembership{},~F. Javier~L\'opez-Mart\'inez, \IEEEmembership{Senior Member}, \textit{IEEE},~and Kai-Kit~Wong,~\IEEEmembership{Fellow}, \textit{IEEE}
	}
	\maketitle
	\begin{abstract}
We provide an analytical characterization of the coverage region of simultaneously transmitting and reflecting reconfigurable intelligent surface (STAR-RIS)-aided two-user downlink communication systems. The cases of orthogonal multiple access (OMA) and non-orthogonal multiple access (NOMA) are considered, under the energy-splitting (ES) protocol. % Specifically, by exploiting Gamma distribution to approximate the equivalent fading channels, we derive closed-form expressions of the coverage region in both transmission and reflection spaces provided by STAR-RIS under energy-splitting (ES) protocol. %Results reveal that the use of STAR-RISs is beneficial to extend the smart radio environment (SRE) and improve the coverage region for the NOMA scenario. 
 Results confirm that the use of STAR-RISs is beneficial to extend the coverage region, and that the use of NOMA provides a better performance compared to the OMA counterpart.% additional smart radio environment and provides a larger coverage region for the NOMA scenario compared with the OMA scheme.
	\end{abstract}
	\begin{IEEEkeywords}
		Coverage region, reconfigurable intelligent surface, non-orthogonal multiple access, energy-splitting 
	\end{IEEEkeywords}%\vspace{-3.5ex}
	\maketitle
	%\blfootnote{\noindent Copyright (c) 2015 IEEE. Personal use of this material is permitted. However, permission to use this material for any other purposes must be obtained from the IEEE by sending a request to pubs-permissions@ieee.org.} %Manuscript received January 25, 2021; revised XXX. The review of this paper was coordinated by XXXX. } 
\blfootnote{\noindent %Manuscript received XX May, 2023. 
 This work was funded in part by Junta de Andaluc\'ia through grant EMERGIA20-00297, and in part by MCIN/AEI/10.13039/501100011033 through grant PID2020-118139RB-I00. %The review of this paper was coordinated by XXXX.
	}
	
	 	\blfootnote{\noindent Farshad Rostami Ghadi and F. Javier L\'opez-Mart\'inez are with the Communications and Signal Processing Lab, Telecommunication Research Institute (TELMA), Universidad de M\'alaga, M\'alaga, 29010, (Spain). F. Javier L\'opez-Mart\'inez is also with the Dept. Signal Theory, Networking and Communications, University of Granada, 18071, Granada (Spain) (e-mail: $\rm farshad@ic.uma.es$, $\rm fjlm@ugr.es$).}

   \blfootnote{\noindent Kai-Kit Wong is with the Department of Electronic and Electrical Engineering, University College London, London WC1E 6BT, United Kingdom and with Yonsei Frontier Lab, Yonsei University, Seoul, 03722, Korea.(e-mail: kai-kit.wong@ucl.ac.uk).}
	 		
	 		 %with Departmento de Ingenieria  de Comunicaciones, Universidad de M\'alaga - Campus de Excelencia Internacional Andaluc\'ia Tech., M\'alaga 29071, Spain  (e-mail: $\rm farshad@ic.uma.es$)}
	 	
	 	%\blfootnote{\noindent  W.P. Zhu is with Department of Electrical and Computer Engineering, Concordia University, Montreal, QC H3G 1M8, Canada. (e-mail: $\rm weiping@ece.concordia.ca$).}
	
	\blfootnote{Digital Object Identifier 10.1109/XXX.2021.XXXXXXX}
	%\IEEEpeerreviewmaketitle
	\vspace{0mm}
	\section{Introduction}\label{sec-intro}
	Reconfigurable intelligent surfaces (RISs) have been recently introduced as promising approaches to enhance the performance of coverage and energy efficiency in sixth-generation (6G) wireless communications \cite{basar2019wireless,basharat2021reconfigurable}. Generally speaking, RISs are engineered structures equipped with a large number of low-cost passive reflecting elements, which can intelligently reconfigure the wireless propagation environment by adjusting the electromagnetic response of each reflecting element, and thus enable a smart radio environment (SRE). Compared to multi-antenna and relay communication systems, RISs do not require dedicated radio frequency (RF) chains, which reduces energy consumption and hardware costs \cite{wu2021intelligent}. 
	
	%Given the aforesaid advantages of RISs, a wide range of industrial and academic research efforts have been carried out to investigate the performance of RIS-aided wireless communications over various scenarios. However, existing contributions mainly focus on the reflecting-only RIS, where transmitters and receivers ought be located on the same side of the RIS, i.e., within the same \textit{half-space} of the SRE, which significantly limits the performance of using RIS. 
 
 One of the key limitations of RISs is related to the need of transmitters and receivers to be located on the same side of the RIS, i.e., within the same \textit{half-space} of the SRE. To further extend the potential of RISs to the entire \textit{full-space}, the concept of simultaneously transmitting and reflecting RISs (STAR-RISs) has been proposed \cite{liu2021star}. %STAR-RISs can realize a highly flexible \textit{full-space} SRE by simultaneously transmitting and reflecting the incident signals \cite{liu2021star,mu2021simultaneously}. 
 The use of STAR-RIS enables a 360$^\circ$ coverage region of SREs, which implies that separate reflection and refraction passive beamforming can be designed for the different serving areas \cite{zhang2022intelligent}. %Furthermore, in order to enhance this full-space coverage, three implementation strategies namely mode switching (MS), energy splitting (ES), and time switching (TS) were proposed, where each operation mode has its own advantages and disadvantages.

Given the appealing features of STAR-RIS, there has been a growing interest to combine these devices with non-orthogonal multiple access (NOMA) techniques \cite{aldababsa2021simultaneous}, with a special focus on the two-user case \cite{zhang2022star,yue2022simultaneously,wang2022outage,wu2021coverage} as a key building block in communication theory. With the aim of maximizing the sum rate, the authors in \cite{aldababsa2021simultaneous} analyzed the STAR-RIS-assisted NOMA under the mode-switching (MS) protocol. The outage probability and diversity gain analysis for STAR-RIS-aided NOMA system with perfect and imperfect successive interference cancellation (SIC) were studied in \cite{zhang2022star} and \cite{yue2022simultaneously}. In addition, by assuming fading channel correlation, the authors in \cite{wang2022outage} derived the outage probability for a STAR-RIS-aided NOMA network. Moreover, a coverage range maximization problem over a STAR-RIS-assisted NOMA network under a quality of service (QoS) requirement for each user was formulated in \cite{wu2021coverage}. While this latter work showed the benefits of NOMA over conventional OMA in terms of coverage range extension from an optimization perspective, the performance characterization in terms of the achievable average rates is challenging due to the involved analytical derivations required in this scenario -- even for the two-user case. Thus, motivated by the importance of rigorously analyzing the coverage area in STAR-RIS, we derive analytical compact expressions for the coverage regions for both transmitting and reflecting users over a STAR-RIS-aided downlink NOMA/OMA, by exploiting the definition of coverage region provided in \cite{aggarwal2009maximizing}. %, where the equivalent channels observed at each user are approximated with Gamma distributions. 
%
% in realizing high spectrum efficiency, massive connectivity, and low latency \cite{chen2016application},
% 
%Regarding the superiority of the STAR-RISs in providing full-space coverage and the potential advantages of non-orthogonal multiple access (NOMA) systems in realizing high spectrum efficiency, massive connectivity, and low latency \cite{chen2016application}, STAR-RIS-aided NOMA communication systems have received significant attention recently 
%
%Despite the above-mentioned contributions, the exact analytical expression of coverage region over STAR-RIS-aided NOMA communication systems has not been investigated in previous works to the best of our knowledge. Only recently, the authors in \cite{wu2021coverage} solved a non-convex problem to maximize the coverage range of the STAR-RIS-aided NOMA network, however, they did not provide the exact analytical expression. Thus, motivated by the importance of analyzing coverage area, we derive closed-form expressions of the coverage region for both transmitting and reflecting users over a STAR-RIS-aided downlink NOMA/OMA by exploiting the definition of coverage region provided in \cite{aggarwal2009maximizing}, where the equivalent channels observed at each user are approximated with Gamma distributions. %Numerical results show that the proposed analytical expressions are accurate, and STAR-RIS can significantly extend the NOMA system performance in terms of the coverage region. 
%
Numerical results show that the proposed analytical expressions are rather accurate, and confirms that STAR-RIS can significantly enhance the system performance in terms of the coverage region for the NOMA scenario compared with the OMA scheme. 

The remainder of this paper is organized as follows: Section \eqref{sec-sys} describes the system model of interest, while section \ref{sec-cov} provides the analysis of the coverage region. Numerical results and conclusion are provided in Sections \ref{sec-num} and \ref{sec-con}, respectively.    
 \section{System Model}\label{sec-sys}
	%\subsection{Signal Model}
	We consider a STAR-RIS aided downlink NOMA as shown in Fig. \ref{system}, where a single-antenna base station (BS) serves two single-antenna mobile users (i.e., reflecting user $u_\mathrm{r}$ and transmitting user $u_\mathrm{t}$) aided by a fixed STAR-RIS with $N$ elements. We assume that the direct link between the BS and the reflecting user $u_\mathrm{r}$ is blocked due to obstacles, such as trees or buildings, and thus only reflecting links via the STAR-RIS are considered for the reflecting user $u_\mathrm{r}$. In addition, the transmitting user $u_\mathrm{t}$ is located behind the STAR-RIS, thereby the direct link for the transmitting user $u_\mathrm{t}$ is blocked. Thus, the received signal at the user $u_\mathrm{k}$, $\mathrm{k}\in\left\{\mathrm{r},\mathrm{t}\right\}$ is expressed as
	\begin{align}%\beta_\mathrm{t}
	Y_\mathrm{t}=\textcolor{black}{}\sqrt{P}\vec{h}^{T}\vec{\Phi}\vec{g}_\mathrm{t}\left(\textcolor{black}{\sqrt{p_\mathrm{t}}}X_\mathrm{t}+\textcolor{black}{\sqrt{p_\mathrm{r}}}X_\mathrm{r}\right)+Z_\mathrm{t}
	\end{align}
	\begin{align}%\beta_\mathrm{r}
	Y_\mathrm{r}=\textcolor{black}{}\sqrt{P}\vec{h}^{T}\vec{\Psi}\vec{g}_\mathrm{r}\left(\textcolor{black}{\sqrt{p_\mathrm{t}}}X_\mathrm{t}+\textcolor{black}{\sqrt{p_\mathrm{r}}}X_\mathrm{r}\right)+Z_\mathrm{r}
\end{align}
where $P$ is the total transmit power, $X_\mathrm{k}$ denotes the symbol transmitted to user $u_\mathrm{k}$ with unit power  (i.e., $\mathbb{E}\left[|X_\mathrm{k}|^2\right]=1$), $p_\mathrm{k}$ is the power allocation factor for $u_\mathrm{k}$, so that $p_t+p_r=1$, and $Z_\mathrm{k}$ is the additive white Gaussian noise (AWGN) with zero mean and variance $\sigma^2$ at $u_\mathrm{k}$. We consider the energy-splitting (ES)\footnote{These results can extended to the mode-switching (MS) and time-switching (TS) modes by allocating different numbers of active STAR-RIS elements or time blocks, respectively.} protocol for our proposed STAR-RIS-aided network. Hence, we assume that all elements of the STAR-RIS simultaneously operate refraction and reflection modes, while the total radiation energy is split into two parts, i.e, $\vec{\Phi}=\text{diag}\left(\left[\textcolor{black}{\beta_\mathrm{t,1}}\mathrm{e}^{j\phi_1},\textcolor{black}{\beta_\mathrm{t,2}}\mathrm{e}^{j\phi_2},...,\textcolor{black}{\beta_{\mathrm{t},N}}\mathrm{e}^{j\phi_{N}}\right]\right)$ and  $\vec{\Psi}=\text{diag}\left(\left[\textcolor{black}{\beta_\mathrm{r,1}}\mathrm{e}^{j\psi_1},\textcolor{black}{\beta_\mathrm{r,2}}\mathrm{e}^{j\psi_2},...,\textcolor{black}{\beta_{\mathrm{r},N}}\mathrm{e}^{j\psi_{N}}\right]\right)$ denote the STAR-RIS transmitting and reflecting coefficients matrices, respectively, where $\phi_n$ and $\psi_n$ are the adjustable phases induced by the $n$th element of the STAR-RIS during transmission and reflection, whereas \textcolor{black}{$\beta_{\mathrm{k},n}$} denote the adjustable transmission/reflection coefficients, with $\beta^2_{\mathrm{r},n}+\beta^2_{\mathrm{t},n}\leq1$ due to the passive nature of the STAR-RIS. In order to minimize the system complexity, we assume in the sequel that all elements have the same amplitude coefficients \cite{wu2021coverage}, i.e., \textcolor{black}{$\beta_{\mathrm{k},n}=\beta_\mathrm{k},\,\forall{n=1\ldots N}$}. The vectors $\vec{h}^T=d^{-\alpha}\left[h_1\mathrm{e}^{-j\theta_1}, h_2\mathrm{e}^{-j\theta_2}, ..., h_{N}\mathrm{e}^{-j\theta_{N}}\right]$, $\vec{g}_\mathrm{t}=d_\mathrm{t}^{-\alpha}\left[g_{\mathrm{t},1}\mathrm{e}^{-j\zeta_1},g_{\mathrm{t},2}\mathrm{e}^{-j\zeta_2},...,g_{\mathrm{t},N}\mathrm{e}^{-j\zeta_{N}}\right]$, and  $\vec{g}_\mathrm{r}=d_\mathrm{r}^{-\alpha}\left[g_{\mathrm{r},1}\mathrm{e}^{-j\eta_1},g_{\mathrm{r},2}\mathrm{e}^{-j\eta_2},...,g_{\mathrm{r},N}\mathrm{e}^{-j\eta_{N}}\right]$ contain the channel gains from the BS to each element of the STAR-RIS, the channel gains from each element of the STAR-RIS to $u_\mathrm{t}$, and the channel gains from each element of the STAR-RIS to $u_\mathrm{r}$, respectively. The terms $d$, $d_\mathrm{t}$, and $d_\mathrm{r}$ define the distances between the BS and the STAR-RIS, the distance between the STAR-RIS and $u_\mathrm{t}$, and the distance between the STAR-RIS and $u_\mathrm{r}$, respectively, where $\alpha$ is the path-loss exponent. In addition, $h_n$, $g_{\mathrm{t},n}$, and $g_{\mathrm{r},n}$ are the amplitudes of the corresponding channel gains, and $\mathrm{e}^{-j\theta_n}$, $\mathrm{e}^{-j\zeta_n}$, and $\mathrm{e}^{-j\eta_n}$ represent
the phases of the respective links.  %In the other words, to reach the best utilization rate of STAR-RIS elements, the sum energy of the transmitted and reflected signals has to be equal to that of the incident signals, i.e., $\beta^2_\mathrm{r}+\beta^2_\mathrm{t}=1$ \cite{mu2021simultaneously}.
	\begin{figure}[!t]
	\centering
	\includegraphics[width=0.9\columnwidth]{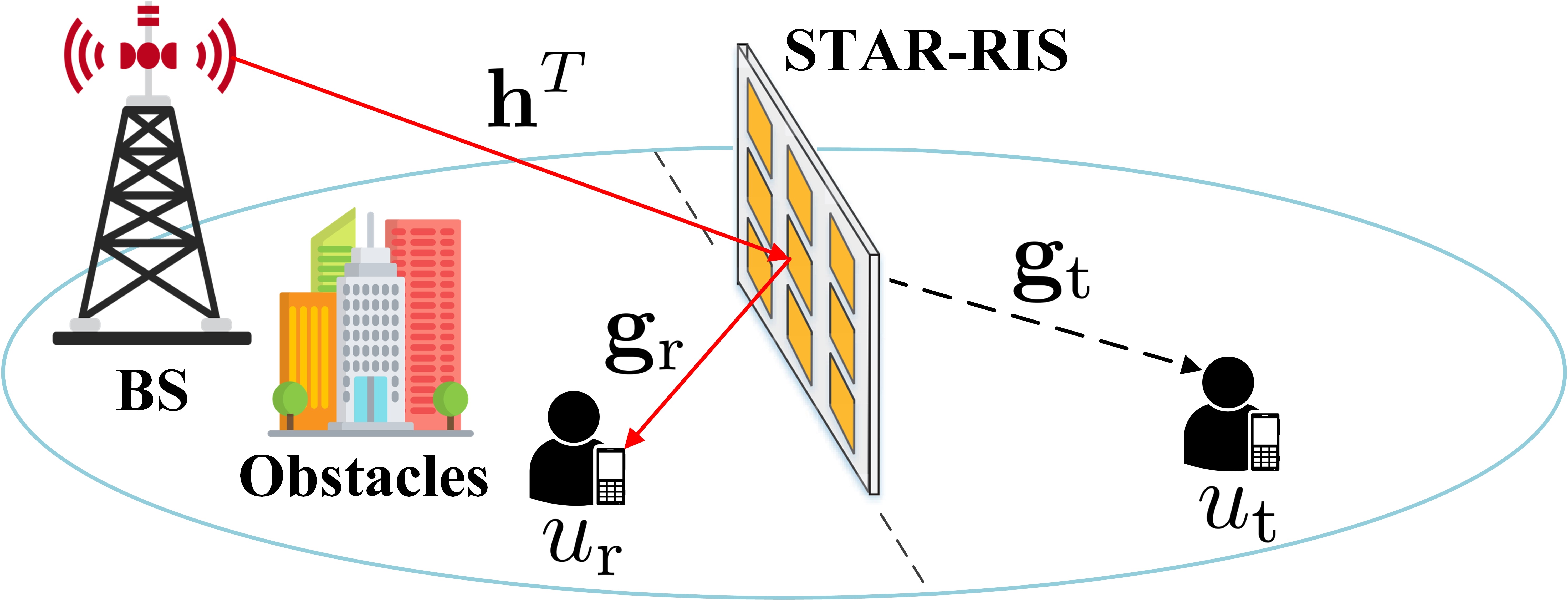}
	\caption{System model of a downlink STAR-RIS-aided NOMA communication.}
	\label{system}
\end{figure}
	%\subsection{SNR Distribution}
	
	As per the principles of NOMA, the BS transmits the signals of both reflecting and transmitting users using the same time and frequency resources by  superposition coding. Besides, the NOMA user with a better channel condition conducts successive interference cancellation (SIC), while the other user decodes its signal directly treating interference as noise. %However, it should be noted that the analysis of optimal decoding order is too complicated since the channel conditions of two users fluctuate quickly. Therefore, 
Without loss of generality, we assume $d_\mathrm{r}<d_\mathrm{t}$ representing the scenario on which a STAR-RIS is deployed to overcome blockages for reflection users (see Fig. \ref{system}). For the sake of notational simplicity yet without loss of generality, we assume a fixed decoding order based on path loss \cite{zhang2022star,chen2016application,yue2022simultaneously,zhu2020power} so that the strong user $u_\mathrm{r}$ operates the SIC process. This implies in turn that the user in the transmission plane $u_\mathrm{t}$ is allocated more power by the BS (i.e., $p_{\mathrm{r}}<p_{\mathrm{t}}$). %This implies the use of fixed decoding order based on path-loss definition, which is coincident with the considerations taken in previous works . %By doing so, to ensure that the strong user $u_\mathrm{r}$ operates the SIC process, more transmit power allocated to the weak user $u_\mathrm{t}$ by the BS. 
Thus, the signal-to-interference-plus-nois ratio (SINR) of the SIC process for $u_\mathrm{r}$ can be defined as
	\begin{align}
	\gamma_\mathrm{sic}=\frac{\rho\textcolor{black}{p_{\mathrm{r}}\beta_{\mathrm{r}}^2}d^{-\alpha}d_{\mathrm{r}}^{-\alpha}\left|\sum_{n=1}^Nh_ng_{\mathrm{r},n}\mathrm{e}^{j\left(\phi_n-\theta_n-\zeta_n\right)}\right|^2}{\rho\textcolor{black}{p_{\mathrm{t}}\beta_{\mathrm{r}}^2}d^{-\alpha}d_{\mathrm{r}}^{-\alpha}\left|\sum_{n=1}^Nh_ng_{\mathrm{r},n}\mathrm{e}^{j\left(\phi_n-\theta_n-\zeta_n\right)}\right|^2+1},
	\end{align}
where $\rho=P/\sigma^2$ is transmit signal-to-noise ratio (SNR). Then, with the aid of SIC, $u_\mathrm{r}$ removes the message of $u_\mathrm{t}$ from its received signal and decodes its required information with the following SNR
\begin{align}
\gamma_\mathrm{r}=\rho d^{-\alpha}d_{\mathrm{r}}^{-\alpha}\textcolor{black}{p_{\mathrm{r}}\beta_{\mathrm{r}}^2}\left|\sum_{n=1}^Nh_ng_{\mathrm{r},n}\mathrm{e}^{j\left(\phi_n-\theta_n-\zeta_n\right)}\right|^2.
\end{align}
At the same time, $u_\mathrm{t}$ directly decodes its signal by considering the signal of $u_\mathrm{r}$ as interference. Thus, the SINR at $u_\mathrm{t}$ can be expressed as
	   \begin{align}
	   	\gamma_\mathrm{t}=\frac{\rho\textcolor{black}{p_{\mathrm{t}}\beta_{\mathrm{t}}^2}d^{-\alpha}d_{\mathrm{t}}^{-\alpha}\left|\sum_{n=1}^Nh_ng_{\mathrm{t},n}\mathrm{e}^{j\left(\psi_n-\theta_n-\eta_n\right)}\right|^2}{\rho \textcolor{black}{p_{\mathrm{r}}\beta_{\mathrm{t}}^2}d^{-\alpha}d_{\mathrm{t}}^{-\alpha}\left|\sum_{n=1}^Nh_ng_{\mathrm{t},n}\mathrm{e}^{j\left(\psi_n-\theta_n-\eta_n\right)}\right|^2+1}.
	   \end{align}
   
   On the other hand, the SNR at user $u_\mathrm{k}$ in OMA, such as time division multiple access (TDMA), can be respectively defined as
   \begin{align}
   	\gamma^{\mathrm{o}}_\mathrm{r}=\rho d^{-\alpha}d^{-\alpha}_\mathrm{r}\textcolor{black}{\beta^2_\mathrm{r}}\left|\sum_{n=1}^Nh_ng_{\mathrm{r},n}\mathrm{e}^{j\left(\phi_n-\theta_n-\zeta_n\right)}\right|^2,
\end{align}
   \begin{align}
	\gamma^{\mathrm{o}}_\mathrm{t}=\rho d^{-\alpha}d^{-\alpha}_\mathrm{t}\textcolor{black}{\beta^2_\mathrm{t}}\left|\sum_{n=1}^Nh_ng_{\mathrm{t},n}\mathrm{e}^{j\left(\textcolor{black}{\psi_n}-\theta_n-\eta_n\right)}\right|^2.
\end{align}
%\textcolor{green}{When using OMA, we allocate all power to user $r$ when transmitting $X_r$ using all power available $P$. Then, in the next block we allocate all power to user $t$ when transmitting $X_t$ using all power available $P$. So, we are either using twice as much resources (time slots) or power. My point is, in order to enable a fair comparison, should be modify $P$ in the OMA scenario? Maybe, for coherence, the easiest thing to do is to include a coefficient to denote the proportion of resource blocks allocated to user k, as in \cite{wu2021coverage}.}

%\textcolor{green}{Note that I modified some equations in the system model, but didn't in Section III. These need to be updated accordingly once we agree on the contents of Section II}

Aiming to maximize the SNR at both reflecting and transmitting users, and as a benchmark for upper bounding the achievable performances, the ideal phase shifting is considered in the proposed system model \cite{zhang2022star}. Without loss of generality, we assume that all channels undergo Rayleigh fading. Even in this case, the equivalent channels at both users include a sum of products of random variables (RVs), which does not admit a tractable probability density function (PDF) expression. For this purpose, we exploit a Gamma approximation to derive equivalent channel distributions, which is reportedly more accurate than the central limit theorem (CLT) approximation \cite{atapattu2020reconfigurable}. Thus, by defining $W_\mathrm{k}=\left|\sum_{n=1}^Nh_ng_{\mathrm{k},n}\right|^2$, the PDF of $W_\mathrm{k}$ %$\mathrm{k}\in\{\mathrm{t},\mathrm{r}\}$ 
can be approximated as:
   \begin{align}
	f_{W_\mathrm{k}}\left(w_\mathrm{k}\right)=\frac{w_\mathrm{k}^{\kappa-1}\mathrm{e}^{-\frac{w_\mathrm{k}}{\tau}}}{\Gamma\left(\kappa\right)\tau^\kappa},
\end{align}
%\begin{align}
%	F_{W_\mathrm{k}}\left(w_\mathrm{k}\right)=1-\frac{\Gamma\left(\kappa,\frac{w_\mathrm{k}}{\tau}\right)}{\Gamma\left(\kappa\right)},
%\end{align}
%   \begin{align}
%   	f_{W_\mathrm{k}}\left(w_\mathrm{k}\right)=\frac{w_\mathrm{k}^{\frac{\kappa}{2}-1}\mathrm{e}^{-\frac{w_\mathrm{k}}{2\tau}}}{2\Gamma\left(\kappa\right)\tau^\kappa},
%   \end{align}
%\begin{align}
%F_{W_\mathrm{k}}\left(w_\mathrm{k}\right)=1-\frac{\Gamma\left(\kappa,\frac{\sqrt{w}_\mathrm{k}}{\tau}\right)}{\Gamma\left(\kappa\right)},
%\end{align}
where $\kappa=\frac{N\pi^2}{16-\pi^2}$ and $\tau=\frac{16-\pi^2}{4\pi}$ denote the shape and scale parameters, respectively.

\section{Coverage Region Analysis}\label{sec-cov}
We now derive the analytical expressions of the coverage regions for both reflecting and transmitting users. To this end, we exploit the concept of coverage region provided in \cite[Def. 2]{aggarwal2009maximizing}. %, which is closely related to the concept of outage capacity \cite{host2002capacity}.
Without loss of generality, we assume that the BS is located at the origin of the two-dimensional coordinate and the STAR-RIS is placed at the fixed distance $d$ from the BS. Under such an assumption, the coverage region can be defined as the geographic area for which the rates $R_\mathrm{r}>0$ and $R_\mathrm{t}>0$ are guaranteed, i.e., 
\begin{align}\label{eq-cov-def}
	\mathcal{G}\left(d_\mathrm{t},d_\mathrm{r}\right)\overset{\Delta}{=}\left\{d_\mathrm{t},d_\mathrm{r},C_{\mathrm{t}}\left(d_\mathrm{t}\right)> R_\mathrm{t},C_{\mathrm{r}}\left(d_\mathrm{r}\right)> R_\mathrm{r}\right\},
\end{align} 
where $C_\mathrm{t}\left(d_\mathrm{t}\right)=\mathbb{E}_{\gamma_\mathrm{t}}\left[\log_2\left(1+\gamma_\mathrm{t}\right)\right]$ and $C_\mathrm{r}\left(d_\mathrm{r}\right)=\mathbb{E}_{\gamma_\mathrm{r}}\left[\log_2\left(1+\gamma_\mathrm{r}\right)\right]$ denote the channel capacity when the users $u_\mathrm{t}$ and $u_\mathrm{t}$ are located at distances $d_\mathrm{t}$ and $d_\mathrm{r}$, respectively. 
\begin{remark}
The coverage region in \eqref{eq-cov-def} refers to the maximum distance $d_\mathrm{k}$ that user $u_\mathrm{k}$ can have from the STAR-RIS as far as reliable communication is guaranteed. This definition is related to an expectation over the achievable rate classically associated to an ergodic setting, which is also relevant in a block-fading setting with link adaptation capabilities \cite{lozano2012yesterday}.
	\end{remark}
%
 %\textcolor{green}{I am deleting the $\frac{1}{2}$ from the instantaneous rate definition in NOMA, similarly to how is done in \cite{wu2021coverage}.}
\begin{theorem}\label{thm-cov-n}
The coverage region  over a STAR-RIS-aided NOMA communication for users $u_\mathrm{r}$ and $u_\mathrm{t}$ with defined parameters $R^{\mathcal{N}}_\mathrm{k}$, $\tau$, $\kappa$, $\rho$, and $\beta_{\mathrm{k}}$ is respectively given by 
\begin{align}\label{cov-r}
	R^\mathcal{N}_\mathrm{r}\leq\frac{1}{\Gamma\left(\kappa\right)\ln 2}G_{3,2}^{1,3}\left(\begin{array}{c}
		\frac{\tau\rho\textcolor{black}{p_\mathrm{r}}\beta_{\mathrm{r}}^2}{d^\alpha d_\mathrm{r}^\alpha}\end{array}
	\Bigg\vert\begin{array}{c}
		\left(1-\kappa,1,1\right)\\
		(1,0)\\
	\end{array}\right),
\end{align}
\begin{align}\nonumber
	R^\mathcal{N}_\mathrm{t}\leq&\,\frac{1}{\Gamma(\kappa)\ln2}G_{3,2}^{1,3}\left(\begin{array}{c}
		\frac{\tau\rho\beta_{\mathrm{t}}^2}{d^\alpha d_\mathrm{t}^\alpha}\end{array}
	\Bigg\vert\begin{array}{c}
		\left(1-\kappa,1,1\right)\\
		(1,0)\\
	\end{array}\right)\\
	&-\frac{1}{\Gamma(\kappa)\ln2}G_{3,2}^{1,3}\left(\begin{array}{c}
		\frac{\tau\rho\textcolor{black}{p_\mathrm{r}}\beta_{\mathrm{t}}^2}{d^\alpha d_\mathrm{t}^\alpha}\end{array}
	\Bigg\vert\begin{array}{c}
		\left(1-\kappa,1,1\right)\\
		(1,0)\\
	\end{array}\right),
\end{align}
where $G^{m,n}_{p,q}(.)$ denotes the Meijer's G-function. 
\end{theorem}
\begin{proof}
%Following the definition of outage capacity for reliable communication, 
The coverage region is expressed in terms of an expectation over the random equivalent channel $w_\mathrm{k}$. Thus, we first determine the coverage region of $u_\mathrm{r}$ as follows:
	\begin{align}
	&R^\mathcal{N}_\mathrm{r}\leq\mathbb{E}_{w_\mathrm{r}}\left[\log_2\left(1+\frac{\rho\textcolor{black}{p_\mathrm{r}} \beta_{\mathrm{r}}^2w_\mathrm{r}}{d^{\alpha}d_{\mathrm{r}}^{\alpha}}\right)\right]\\
	&=\int_0^\infty \log_2\left(1+\frac{\rho \textcolor{black}{p_\mathrm{r}}\beta_{\mathrm{r}}^2w_\mathrm{r}}{d^{\alpha}d_{\mathrm{r}}^{\alpha}}\right)f_{W_\mathrm{r}}(w_\mathrm{r})dw_\mathrm{r}\\
	&=\frac{1}{\Gamma(\kappa)\tau^\kappa}\int_0^\infty w_\mathrm{r}^{\kappa-1}\mathrm{e}^{-\frac{w_\mathrm{r}}{\tau}} \log_2\left(1+\frac{\rho \textcolor{black}{p_\mathrm{r}}\beta_{\mathrm{r}}^2w_\mathrm{r}}{d^{\alpha}d_{\mathrm{r}}^{\alpha}}\right)dw_\mathrm{r},
\end{align}
where, by assuming $v_\mathrm{r}=\frac{\rho\textcolor{black}{p_\mathrm{r}} \beta^2_\mathrm{r}}{d^{\alpha}d_{\mathrm{r}}^{\alpha}}w_{\mathrm{r}}$ %\Rightarrow \mathrm{d}w_\mathrm{r}=\frac{d^\alpha d^\alpha_\mathrm{r}}{\rho\beta^2_\mathrm{r}}\mathrm{d}u_{\mathrm{r}}$ 
 and expressing the logarithm in terms of Meijer's G-function \cite{prudnikov1990more}, we have
\begin{align}\nonumber
	R^\mathcal{N}_\mathrm{r}&\leq\,\frac{1}{\tau^\kappa\Gamma(\tau)\ln 2}\left(\frac{d^\alpha d^{\alpha}_\mathrm{r}}{\rho\textcolor{black}{p_\mathrm{r}}\beta_\mathrm{r}^2}\right)^{\kappa}\\
	&\times\int_{0}^{\infty}v_\mathrm{r}^{\kappa-1}\mathrm{e}^{-\frac{d^\alpha d^\alpha_\mathrm{r}}{\rho\textcolor{black}{p_\mathrm{r}}\tau\beta_{\mathrm{r}}^2}v_\mathrm{r}}G_{2,2}^{1,2}\left(\begin{array}{c}
		v_\mathrm{r}\end{array}
	\Bigg\vert\begin{array}{c}
		(1,1)\\
		(1,0)\\
	\end{array}\right)\mathrm{d}
	v_\mathrm{r},
\end{align}
now, with the help of \cite[Eq. 2.24.3.1]{prudnikov1990more}, the above integral can be solved and \eqref{cov-r} is obtained. Similarly, the coverage region for user $u_\mathrm{r}$ can defined as
%\begin{align}
%	R_\mathrm{r}\leq\frac{1}{2\Gamma\left(\kappa\right)\ln 2}G_{3,2}^{1,3}\left(\begin{array}{c}
%		\frac{\tau\rho\beta_{\mathrm{r}}^2}{d^\alpha d_\mathrm{r}^\alpha}\end{array}
%	\Bigg\vert\begin{array}{c}
%		\left(1-\kappa,1,1\right)\\
%		(1,0)\\
%	\end{array}\right).
%\end{align}
	\begin{align}
	&R^\mathcal{N}_\mathrm{t}\leq\mathbb{E}_{w_\mathrm{t}}\left[\log_2\left(1+\frac{\rho\textcolor{black}{p_\mathrm{t}} \beta_{\mathrm{t}}^2d^{-\alpha}d^{-\alpha}_\mathrm{t}w_\mathrm{t}}{ \rho\textcolor{black}{p_\mathrm{r}}\beta_{\mathrm{t}}^2d^{-\alpha}d^{-\alpha}_\mathrm{t}w_\mathrm{t}+1}\right)\right]\\
	&\hspace{-1ex}=\hspace{-0.75ex}\int_0^\infty \log_2\left(1+\frac{\rho\textcolor{black}{p_\mathrm{t}} \beta_{\mathrm{t}}^2d^{-\alpha}d^{-\alpha}_\mathrm{t}w_\mathrm{t}}{\rho \textcolor{black}{p_\mathrm{r}}\beta_{\mathrm{t}}^2d^{-\alpha}d^{-\alpha}_\mathrm{t}w_\mathrm{t}+1}\right)f_{W_\mathrm{t}}(w_\mathrm{t})dw_\mathrm{t}\\
	&\hspace{-1ex}=\hspace{-0.75ex}\frac{1}{\Gamma(\kappa)\tau^\kappa}\int_0^\infty \hspace{-2ex}w_\mathrm{t}^{\kappa-1}\mathrm{e}^{-\frac{w_\mathrm{t}}{\tau}} \log_2\hspace{-0.5ex}\left(1+\frac{\rho \textcolor{black}{p_\mathrm{t}}\beta_{\mathrm{t}}^2d^{-\alpha}d^{-\alpha}_\mathrm{t}w_\mathrm{t}}{ \rho\textcolor{black}{p_\mathrm{r}}\beta_{\mathrm{t}}^2d^{-\alpha}d^{-\alpha}_\mathrm{t}w_\mathrm{t}+1}\right)dw_\mathrm{t}\\\nonumber
	&\hspace{-1ex}=\frac{1}{\Gamma(\kappa)\tau^\kappa\ln2}\int_0^\infty  w_\mathrm{t}^{\kappa-1}\mathrm{e}^{-\frac{w_\mathrm{t}}{\tau}}\ln\left(1+\frac{\rho \textcolor{black}{\beta^2_{\mathrm{t}}}w_\mathrm{t}}{d^{\alpha} d^{\alpha}_\mathrm{t}}\right)dw_\mathrm{t}\\
	&\hspace{-1ex}-\frac{1}{\Gamma(\kappa)\tau^\kappa\ln2}\int_0^\infty \hspace{-0.2ex}w_\mathrm{t}^{\kappa-1}\mathrm{e}^{-\frac{w_\mathrm{t}}{\tau}}\ln\left(1+\frac{\rho\textcolor{black}{p_\mathrm{r}}\beta_{\mathrm{t}}^2w_\mathrm{t}}{d^\alpha d^\alpha_\mathrm{t}}\right)dw_\mathrm{t},
	\end{align}
where, by assuming $s_\mathrm{t}=\frac{\rho\textcolor{black}{\beta^2_{\mathrm{t}}}}{d^\alpha d_\mathrm{t}^\alpha}w_\mathrm{t}$,  $v_\mathrm{t}=\frac{\rho\textcolor{black}{p_\mathrm{r}}\beta_{\mathrm{t}}^2}{d^\alpha d_\mathrm{t}^\alpha}w_\mathrm{t}$, and using the Meijer's G-function instead of the logarithm, we have
\begin{align}\nonumber
	R^\mathcal{N}_\mathrm{t}&\leq\,\frac{1}{\Gamma(\kappa)\tau^\kappa\ln2}\left(\frac{d^\alpha d^{\alpha}_\mathrm{t}}{\rho\textcolor{black}{\beta^2_\mathrm{t}}}\right)^{\kappa}\\ \nonumber
	&\times\int_0^\infty s_\mathrm{t}^{\kappa-1}\mathrm{e}^{-\frac{d^\alpha d_\mathrm{t}^\alpha}{\rho\tau\textcolor{black}{\beta^2_\mathrm{t}}}s_\mathrm{t}}G_{2,2}^{1,2}\left(\begin{array}{c}
		s_\mathrm{t}\end{array}
	\Bigg\vert\begin{array}{c}
		(1,1)\\
		(1,0)\\
	\end{array}\right)ds_\mathrm{t}\\\nonumber
&-\frac{1}{\Gamma(\kappa)\tau^\kappa\ln2}\left(\frac{d^\alpha d^{\alpha}_\mathrm{t}}{\rho\textcolor{black}{p_\mathrm{r}}\beta_\mathrm{t}^2}\right)^{\kappa}\\
&\times\int_{0}^{\infty}v_\mathrm{t}^{\kappa-1}\mathrm{e}^{-\frac{d^\alpha d^\alpha_\mathrm{t}}{\rho\textcolor{black}{p_\mathrm{r}}\tau\beta_{\mathrm{t}}^2}v_\mathrm{t}}G_{2,2}^{1,2}\left(\begin{array}{c}
	v_\mathrm{t}\end{array}
\Bigg\vert\begin{array}{c}
	(1,1)\\
	(1,0)\\
\end{array}\right)\mathrm{d}
v_\mathrm{t},
\end{align}
where the above integrals can be calculated with the help of \cite[Eq. 2.24.3.1]{prudnikov1990more} and the proof is completed.
\end{proof}
% \textcolor{green}{I am defining the rate for OMA as in \cite[Eq. 6]{wu2021coverage}. I didn't include $p_k$ in these -- they are included in \cite[Eq. 6]{wu2021coverage}, since we are using TDMA and all power is allocated to the user which has the resources allocated. Please, check it makes sense.}
%
\begin{theorem}\label{thm-cov-o}
	The coverage region  over a STAR-RIS-aided OMA communication for the user $u_\mathrm{k}$ with defined parameters $R^\mathcal{O}_\mathrm{k}$, $\kappa$, $\tau$, and $\rho$ is respectively given by 
	\begin{align}\label{cov-r-o}
		R^{\mathcal{O}}_\mathrm{k}\leq\frac{\textcolor{black}{\delta_\mathrm{k}}}{\Gamma\left(\kappa\right)\ln 2}G_{3,2}^{1,3}\left(\begin{array}{c}
			\frac{\tau\rho\textcolor{black}{p_\mathrm{k}\beta^2_\mathrm{k}}}{\textcolor{black}{\delta_\mathrm{k}}d^\alpha d_\mathrm{k}^\alpha}\end{array}
		\Bigg\vert\begin{array}{c}
			\left(1-\kappa,1,1\right)\\
			(1,0)\\
		\end{array}\right).
\end{align}
\end{theorem}
\begin{proof}
Defining $\textcolor{black}{\delta_\mathrm{k}}\in[0,1]$ as the share of resources allocated to each user in OMA, we have
\begin{align}
    R^\mathcal{O}_\mathrm{k}\leq\mathbb{E}_{w_\mathrm{k}}\left[\textcolor{black}{\delta_\mathrm{k}}\log_2\left(1+\frac{\rho \textcolor{black}{p_\mathrm{k}\beta_{\mathrm{k}}^2w_\mathrm{k}}}{\textcolor{black}{\delta_\mathrm{k}}d^{\alpha}d_{\textcolor{black}{\mathrm{k}}}^{\alpha}}\right)\right].
\end{align}
With these definitions, the proof is directly obtained following similar steps as those
in the proof of Theorem \ref{thm-cov-n}.
\end{proof}\vspace{-0.6cm}
\begin{figure*}[!t]
	\centering
	%\hspace{-0.2cm}
	\subfigure[]{%
		\includegraphics[width=0.34\textwidth]{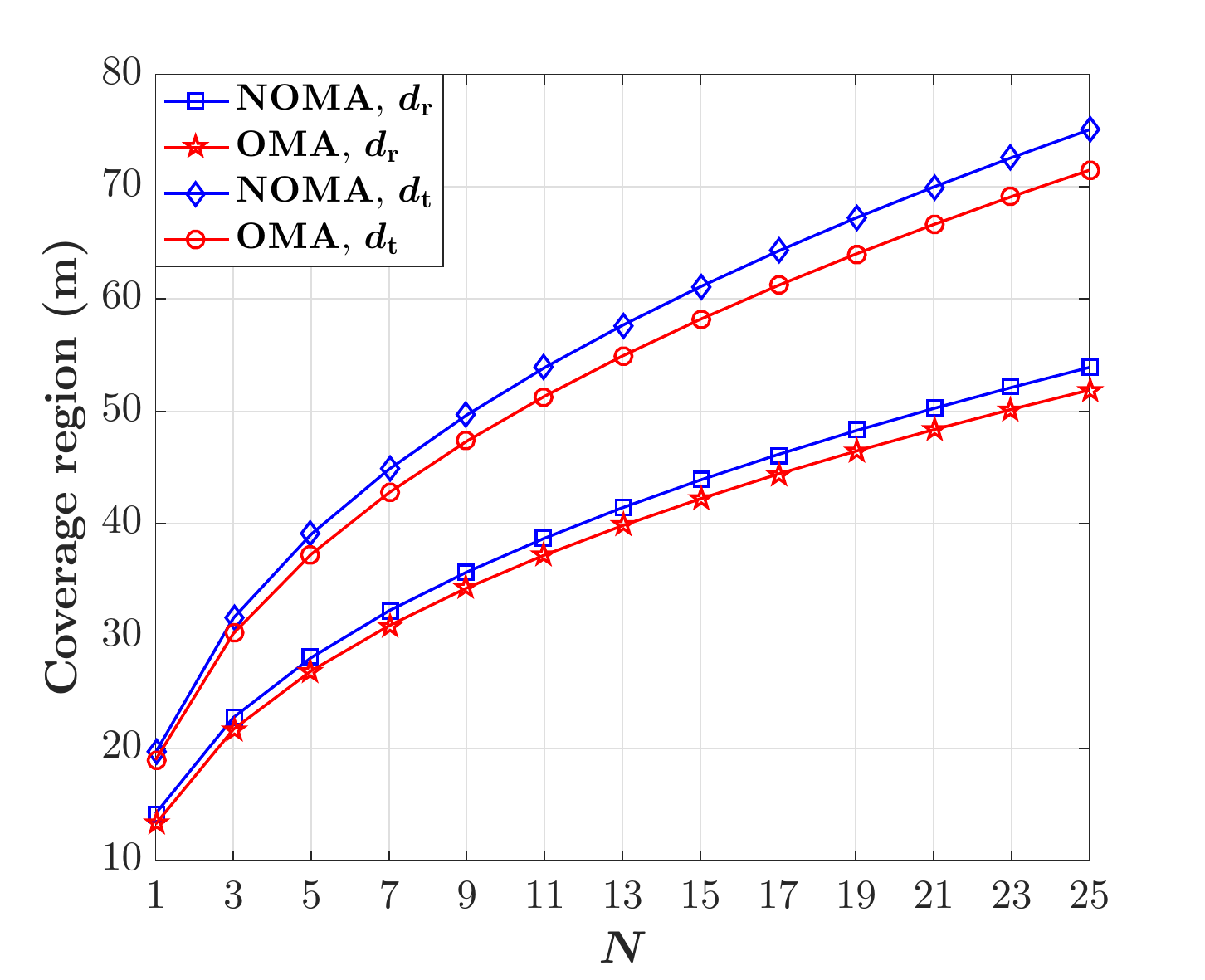}%
		\label{cov_n}%
	}%\hspace{-0.3cm}%or more
	\subfigure[]{%
		\includegraphics[width=0.34\textwidth]{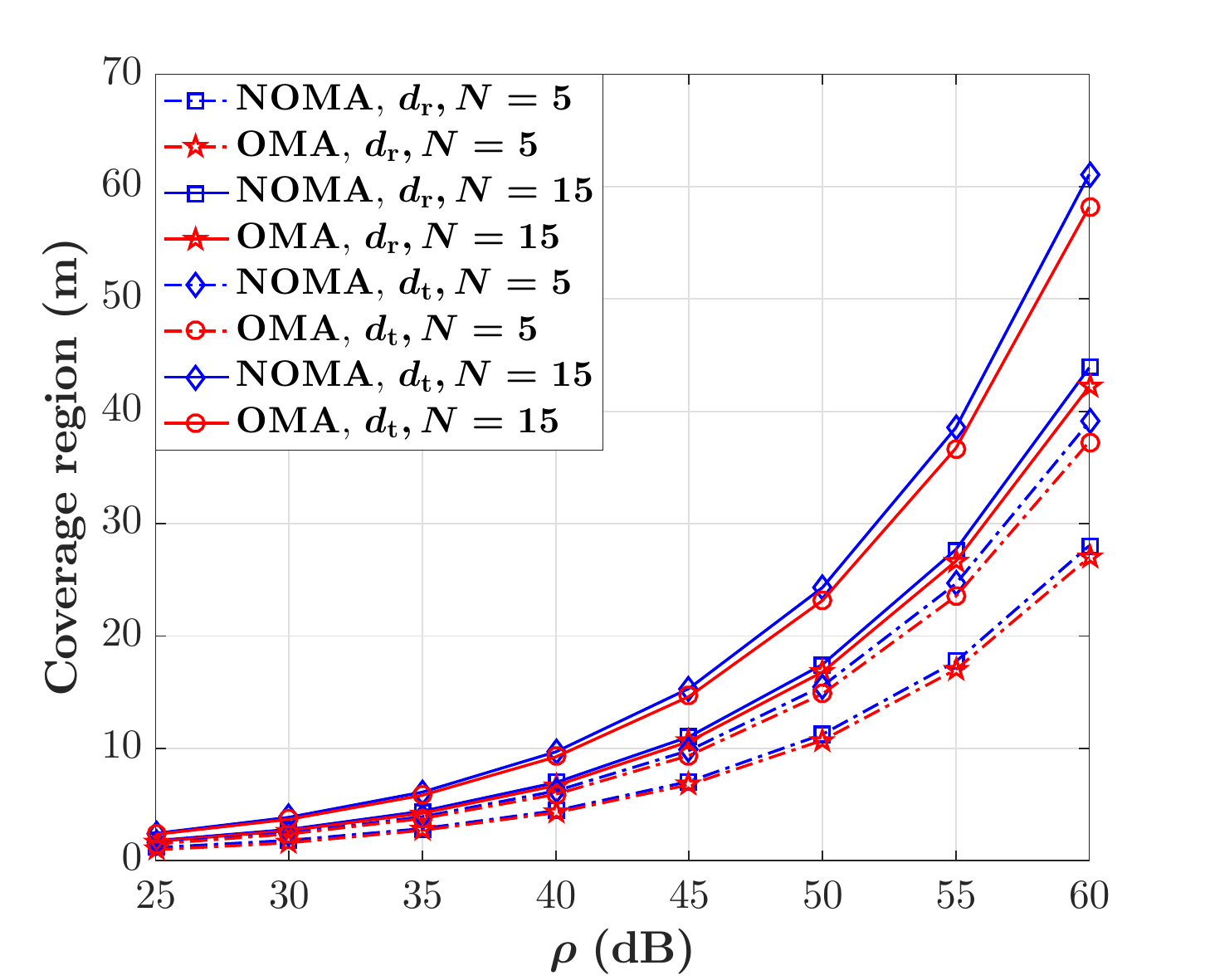}%
		\label{dr_rho}%
	}%\hspace{-0.3cm}%or more
	\subfigure[]{%
		\includegraphics[width=0.34\textwidth]{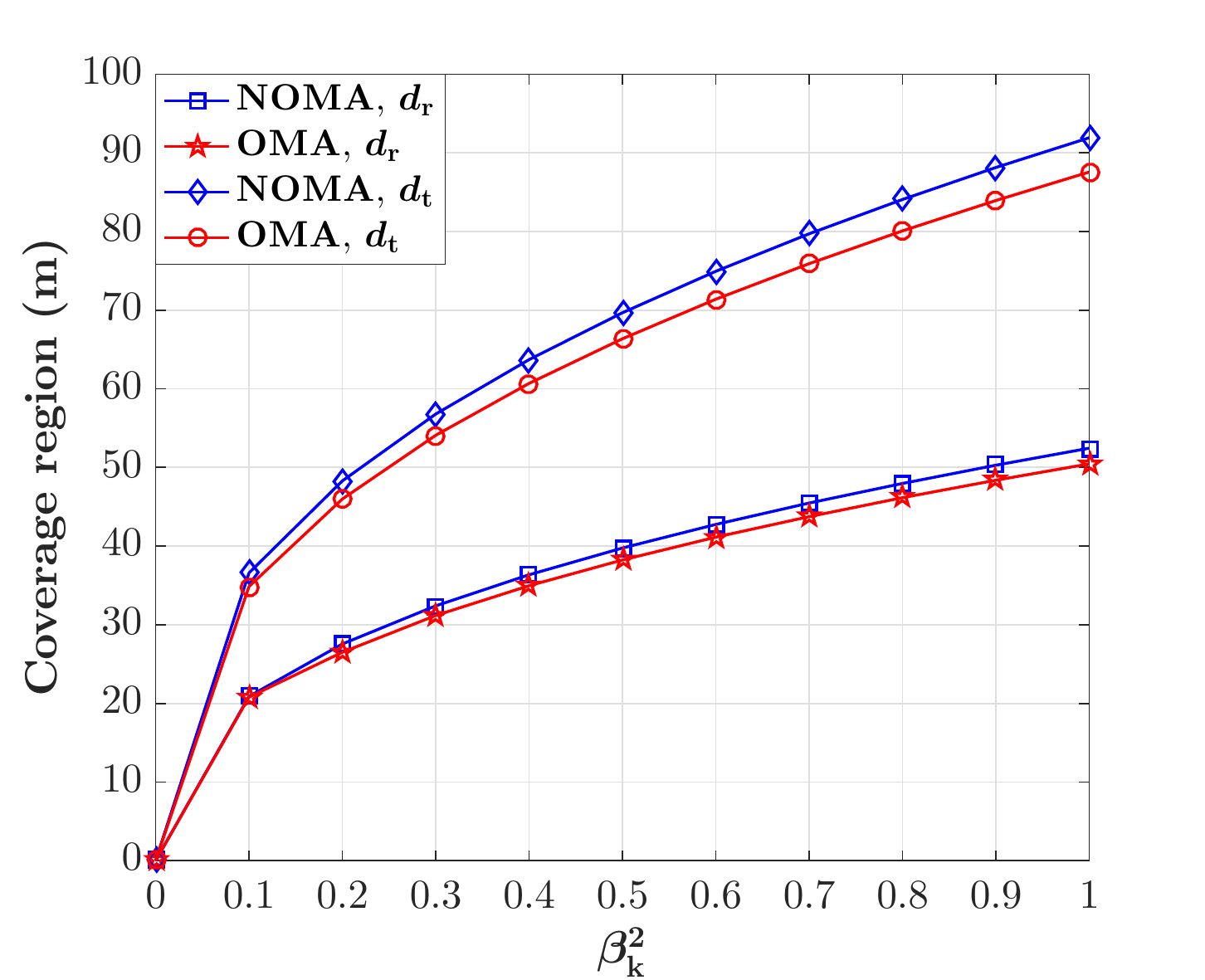}%
		\label{dr_beta}%
	}%\hspace{-0.3cm}%or more
	\caption{The coverage regions $d_\mathrm{r}$ and $d_\mathrm{t}$ versus (a) The number of STAR-RIS element $N$; (b) The transmit SNR $\rho$; and (c) The square of the transmission/reflection coefficients $\beta_\mathrm{k}$.}\vspace{-0.3cm}
	%\label{fig-cov}
\end{figure*}
\section{Numerical Results}\label{sec-num}
%\textcolor{green}{It is possible that some of the figures need to be re-done to accomodate for the new rate definitions, please check.}
We present numerical results to evaluate the analytical expressions previously derived. We assume that the BS is located at the origin of the two-dimensional coordinate and the STAR-RIS with $N$ elements is placed $d=10$ meters away from it. We also set $P=30$dBm, $\sigma^2=-30$dBm, $\beta_{\mathrm{r}}=0.8$,  $\beta_{\mathrm{t}}=0.6$, $p_\mathrm{r}=0.4$, $p_\mathrm{t}=0.6$, $\alpha=2.5$,  $R^\mathcal{K}_\mathrm{r}=0.8$bps/Hz, and  $R^\mathcal{K}_\mathrm{t}=0.3$bps/Hz for $\mathcal{K}\in\{\mathcal{N},\mathcal{O}\}$. In addition, we consider the power allocation factor and the time allocation factor for both users in OMA as $p_\mathrm{k}=0.5$ and $\delta_\mathrm{k}=0.5$, respectively. %\textcolor{green}{We need to specify $p_k$ and $w_k$, noting that we defined $p_t>p_r$ in section II.}

Fig. \ref{cov_n} shows the performance of the coverage region for both reflecting and transmitting users based on the number of STAR-RIS elements $N$. We see that as $N$ increases, the STAR-RIS covers a larger area for users $u_\mathrm{r}$ and $u_\mathrm{t}$. In addition, we can observe that user $u_\mathrm{t}$ can be located in a region further away from the STAR-RIS compared with user $u_\mathrm{r}$, which confirms the adequacy of STAR-RIS to support outdoor users. Moreover, it is observed that the STAR-RIS under NOMA scenario provides a remarkable performance in terms of the coverage region compared with the OMA case since the simultaneously transmitting and reflecting scheme can increase the channel disparity between these two, where NOMA provides better performance than OMA. The behavior of the coverage region for users $u_\mathrm{r}$ and $u_\mathrm{t}$ in terms of the transmit SNR $\rho$ is presented in Fig. \ref{dr_rho}. In this figure, we %can see that by increasing $\rho$, the STAR-RIS can provide a larger region for both users, meaning that $u_\mathrm{r}$ and $u_\mathrm{t}$ can be further away from the BS so that transmission is still reliable. We can also 
observe that for a fixed value of $\rho$, increasing the number of STAR-RIS elements has a more noticeable effect on the variation of distance $d_\mathrm{t}$ compared with $d_\mathrm{r}$. The aforesaid results indicate that applying the STAR-RIS into NOMA networks improves the system performance in terms of the coverage region. \textcolor{black}{It can be seen from Thms. \ref{thm-cov-n} and \ref{thm-cov-o} that the coverage region for each user highly depends on its corresponding reflection/transmission coefficients. In this regard, Fig. \ref{dr_beta} shows the coverage region in terms of the square of the reflection/transmission coefficients $\beta^2_\mathrm{k}$ for a fixed value of $N=15$. It is observed that as $\beta^2_\mathrm{k}$ grows, a larger coverage area can be provided by the STAR-RIS for the respective user $u_\mathrm{k}$. In addition, we can see that under an equal value of reflection and transmission coefficients (i.e., $\beta^2_\mathrm{r}=\beta^2_\mathrm{t}=0.5$), STAR-RIS offers a larger coverage area for user $u_\mathrm{t}$.}\vspace{-0.3cm}

\vspace{0cm}\section{Conclusion}\label{sec-con}\vspace{0cm}
In this paper, we analyzed a downlink STAR-RIS-aided NOMA/OMA communication system, and provided analytical expressions for the coverage regions for users in the transmission and reflection planes of the STAR-RIS under the energy-splitting protocol. We showed that the use of STAR-RIS provides a larger coverage area for \textit{both} transmitting and reflecting users in the NOMA network compared to the benchmark OMA system.\vspace{-0.1cm}%, where numerical results confirm the accuracy of the proposed coverage region expressions.  
%	\appendices
%	\section{Proof of Theorem}\label{}
	
\bibliographystyle{IEEEtran}
\bibliography{sample.bib}

\end{document}